\documentclass[12pt]{amsart}
\usepackage{amsmath}
\usepackage{amssymb}
\usepackage[latin1]{inputenc}
\usepackage[T1]{fontenc}
\usepackage{graphicx}
\newtheorem{definition}{Definition}

\newtheorem{proposition}{Proposition}

\def\proofsec#1{\par\vskip5pt\noindent{\bf#1}\\}

\newcommand{\bra}{\langle}
\newcommand{\ket}{\rangle}

\newcommand{\N}{\mathbb{N}}
\newcommand{\R}{\mathbb{R}}
\newcommand{\Rl}{\mathbb{R}}
\newcommand{\T}{\mathbb{T}}
\newcommand{\Z}{\mathbb{Z}}
\newcommand{\be}{\begin{equation}}
\newcommand{\eeq}{\end{equation}}
\newcommand{\bet}{\begin{equation*}}
\newcommand{\eeqt}{\end{equation*}}
\newcommand{\bea}{\begin{eqnarray}}
\newcommand{\eeqa}{\end{eqnarray}}
\newcommand{\beat}{\begin{eqnarray*}}
\newcommand{\eeqat}{\end{eqnarray*}}

\newcommand{\hil}{\mathcal{H}}

\newcommand{\lh}{\mathcal{B(H)}}
\newcommand{\schtn}[1]{{\mathcal T}_{#1}({\mathcal H})}
\newcommand{\tc}{\schtn{}}

\newcommand{\hs}{\schtn2}
\newcommand{\LL}{{\mathcal L}}
\newcommand{\phaseS}{X}
\newcommand{\Cn}{{\mathcal C}_0}
\newcommand{\kh}{{\mathcal{ K(H)}}}

\newcommand{\G}{\mathsf{G}}

\newcommand{\tr}[1]{\mathrm{tr}\left[ {#1} \right]}

\def\idty{{\leavevmode\rm 1\mkern -5.4mu I}} 
\def\duality#1#2{\langle#1,\,#2\rangle}
\def\zc#1{(Z\,#1)}

\begin{document}
\title{Characterization of informational completeness for covariant phase space observables}

\author{J. Kiukas}
\address{Institute for Theoretical Physics, University of Hannover, Hannover, Germany}
\email{jukka.kiukas@itp.uni-hannover.de}
\author{P. Lahti}
\address{Turku Centre for Quantum Physics, Department of Physics and Astronomy, University of Turku, Turku, Finland}
\email{pekka.lahti@utu.fi}
\author{J. Schultz}
\address{Turku Centre for Quantum Physics, Department of Physics and Astronomy, University of Turku, Turku, Finland}
\email{jussi.schultz@utu.fi}
\author{R. F. Werner}
\address{Institute for Theoretical Physics, University of Hannover, Hannover, Germany}
\email{reinhard.werner@itp.uni-hannover.de}

\begin{abstract}
A covariant phase space observable is uniquely characterized by a positive operator of trace one and, in turn, by the Fourier-Weyl transform of this operator. We study three properties of such observables, and characterize them in terms of the zero set of this transform. The first is informational completeness, for which it is necessary and sufficient that the zero set has dense complement. The second is a version of informational completeness for the Hilbert-Schmidt class, equivalent to the zero set being of measure zero, and the third, known as regularity, is equivalent to the zero set being empty. We give examples demonstrating that all three conditions are distinct. The three conditions are the special cases for $p=1,2,\infty$ of a more general notion of $p$-regularity defined as the norm density of the span of translates of the operator in the Schatten-$p$ class. We show that the relation between zero sets and $p$-regularity can be mapped completely to the corresponding relation for functions in classical harmonic analysis.
\end{abstract}

\maketitle

\section{Introduction}
The significance of informationally complete phase space observables is undisputed from both foundational and practical aspects of quantum mechanics. The wide variety of applications includes such issues as quantization and continuous variable quantum tomography, where these observables provide one of the main tools for state reconstruction procedures (see e.g. \cite{Leonhardt}). In spite of their importance, it has remained an open and somewhat controversial question which of the covariant phase space observables are informationally complete. In this paper the question is finally settled.

For simplicity we take just one degree of freedom, so the phase space is $\phaseS=\Rl^2$, and the associated Hilbert space is $\hil=\LL^2(\Rl)$. Covariant phase space observables associate a probability density $f_\rho(x)$ on $\phaseS$ to every density operator $\rho$ on $\hil$ according to the formula
\begin{equation}\label{density}
    f_\rho(x)=\tr {\rho\,W(x)TW(x)^*},
\end{equation}
where $x=(q,p)$ is a point in phase space, $W(x)$ is the corresponding Weyl operator, and $T$ is a positive operator with trace 1, which characterizes the observable. The key practical question is how we can reconstruct $\rho$ from the measured density $f_\rho$. The observable is called {\it informationally complete} \cite{Prugovecki1977}, if different states can be distinguished, i.e., the map $\rho\mapsto f_\rho$ is injective. The nature of this task, a sort of non-commutative deconvolution,  becomes clearer if we take the symplectic Fourier transform of \eqref{density}. Known properties of Weyl operators (see below) make this
\begin{equation}\label{Fourier}
\int\frac{dq'\,dp'}{2\pi}\ e^{-i(qp'-q'p)}\, f_\rho(q',p')  = \tr{\rho W(q,p)}\ \overline{\tr{TW(q,p)}}.
\end{equation}
Hence we are asked to reconstruct the function $(q,p)\mapsto\tr{\rho W(q,p)}$ given the right hand side of this equation. Clearly, this will work if $\tr{TW(q,p)}$ is non-zero. Given the nature of $\rho$, some zeros can be tolerated, but how many exactly is not obvious. This is the subject of our paper.

Let us introduce the zero set
\begin{equation}\label{Z}
    Z({T}) =\{ x\in\phaseS  \mid \tr{T W(x)}=0\}.
\end{equation}
There are three obvious possibilities to formalize ``smallness'' of $Z(T)$, an algebraic, a measure theoretic and a topological way:
\begin{itemize}
\item[\zc1]  $Z({T})$ is empty,
\item[\zc2] $Z({T})$ is of measure zero,
\item[\zc3] $Z({T})$ contains no open set, i.e., has dense complement.
\end{itemize}
Trivially, \zc1$\Rightarrow$\zc2$\Rightarrow$\zc3. Since the Weyl transform $\tr{\rho W(x)}$ is a continuous function, one immediately sees that the weakest condition \zc3 is sufficient to guarantee the informational completeness of the observable. We will show (Prop.~\ref{propZ} below) that this is also a necessary condition. In addition, we
demonstrate that neither of the obvious implications can be reversed, correcting thus some of the earlier statements claiming the necessity of condition \zc2  \cite{Ali1977,Busch1989,Busch1995,Healy1995}, or even condition \zc1 \cite{DAriano2004}.
Since \zc3 is indeed a necessary condition, this also shows, contrary to some formal state reconstruction formulas
\cite{Wunsche1997}, that not all covariant phase space observables can be used in quantum tomography. For instance, a phase space observable generated by a slit state $T= \vert \varphi\rangle\langle\varphi \vert$, with a compactly supported $\varphi$, is not informationally complete.

\section{Notations and preliminaries}

For the relevant Hilbert space $\hil=\LL^2 (\R)$, we let  $\lh$, $\hs$ and $\tc$,  equipped with the norms $\Vert\cdot\Vert_\infty$, $\Vert\cdot\Vert_2$, and $\Vert\cdot\Vert_1$, denote the spaces of bounded, Hilbert-Schmidt, and trace class operators on $\hil$, respectively.
We also consider the Schatten classes $\schtn p$ for $1\leq p<\infty$ with the norm $\Vert A\Vert_p=\tr{\vert A\vert^p }^{1/p}$ \cite[Ch.~IX.4]{ReedSimon2}. These contain the Hilbert-Schmidt and trace class as $p=2$ and $p=1$. By $\kh$ we denote the space of compact operators, and by $\Cn(\phaseS)$ the space of continuous complex valued functions on $\phaseS$ vanishing at infinity.

Let $Q$ and $P$ denote the standard operators for the phase space coordinates. For concreteness, we can think of position and momentum of a spinless particle confined to move in one dimension, or quadrature components of a one-mode electromagnetic field. (From the point of view of state reconstruction, the latter application is more relevant.) We recall that the Weyl operators  $W(q,p)= e^{i\frac{qp}{2}} e^{-iqP}e^{ipQ}$, $(q,p)\in\phaseS$, constitute an irreducible projective unitary representation of the phase space translations acting in $\LL^2(\R)$.

A covariant phase space observable is a normalized positive operator valued measure, which assigns to each Borel set $M\subset \phaseS$ a positive operator $\G(M)$ such that $\G$ is $\sigma$-additive in the weak*-topology, $\G(\phaseS)=\idty$, and the covariance condition
\begin{equation}\label{covG}
W(q,p) \G (M) W(q,p)^* = \G(M+ (q,p))
\end{equation}
holds for all Borel sets $M$.
Such observables are necessarily \cite{Holevo1979, Werner1984} of the form $\G(M)=\int_{x\in M}dx\ W(x)TW(x)^*$, with a positive trace class operator $T$ with trace one, which is to say that the measurement outcome probability density $f_\rho$ is given by \eqref{density}.

The proofs of our main results rest heavily on quantum harmonic analysis on phase space \cite{Werner1984}. Therefore we present the basic definitions and results needed in the proof. The idea is to extend the definitions of convolution and Fourier transform to combinations of operators and functions. With this convolution $\LL^1(\phaseS)\oplus\tc$ becomes a $\Z_2$-graded commutative Banach algebra, meaning that, for functions $f,g$ and operators $A,B$, $f*g$ and $A*B$ are functions and $f*A=A*f$ is an operator. The associated Fourier transform, which turns the convolution into a product of functions, is the symplectic Fourier transform on the function part, and the Weyl transform $\widehat T(x)=\tr{TW(x)}$ on the operator part. This formalism is helpful to us, because the density $f_\rho$ is then simply the convolution of $\rho$ with $T$, of which \eqref{Fourier} is just the Fourier transform.

One of the classic themes of classical harmonic analysis are the mapping properties of function spaces. In the extended structure this becomes a correspondence theory by which spaces of functions on phase space (assumed to be closed under phase space translations) are associated with spaces of operators. The moral is that while quantum-classical correspondences between individual observables are ``fuzzy'' and generally depend on the choice of some parameters, the correspondence between translation invariant {\it spaces} of functions and operators is canonical. In this paper we just need the instances:
\begin{equation}\label{correspond}
    \begin{array}{rcl}
    \LL^1(\phaseS) & \leftrightarrow & \schtn 1\\
        \LL^p(\phaseS) & \leftrightarrow & \schtn p\\
      \LL^\infty(\phaseS)&\leftrightarrow&\lh \\
      \Cn(\phaseS)&\leftrightarrow&\kh\end{array}
\end{equation}
Here the double arrow indicates that the convolution of an element on one side by an arbitrary trace class operator gives an element on the other side. As customary for functions the convolution is extended here from $\LL^1(\phaseS)\oplus\tc$ to allow one factor from $\LL^\infty(X)$ or $\lh$.

The convolution in some sense was defined above already by describing its Fourier transform. To give a direct definition let us fix some more notation. Throughout the paper we use $x=(q,p)\in\phaseS$ for phase space points and the scaled Lebesgue measure  $dx= (2\pi)^{-1}dq\,dp$.  Let $\alpha_x$ denote the automorphism induced by phase space translations, i.e., $(\alpha_x  f)(y)=f(y-x)$ for $f\in \LL^\infty (\phaseS)$ and $\alpha_x(A)= W(x)AW(x)^*$ for $A\in\lh$. The map $x\mapsto \alpha_x$ is strongly continuous on $\LL^p(\phaseS)$, $\schtn p$ for $1\leq p<\infty$, and on $\Cn(\phaseS)$ and $\kh$. It is weak*-continuous on $\LL^\infty(\phaseS)$ and $\lh$. The phase space inversion of a function is written by a subscript ``$-$'', so $(g_-)(x)=g(-x)$. Its operator analog is
$S_- =\Pi S\Pi$, where $\Pi$ is the parity operator. We can then write the usual convolution of integrable functions in two equivalent ways, which suggest the extensions to trace class operators $A,B$:
\begin{eqnarray}\label{convols}
  (f*g)(y)&=&\int f(x) g(y-x)\, dx = \int f(x)(\alpha_y g_-)(x)\, dx  \nonumber\\
  f*A=A*f&=& \int f(x) \alpha_x(A)\, dx \\
  (A*B)(y)&=&\tr{A\alpha_y(B_-)} \nonumber
\end{eqnarray}
It is a crucial fact of the theory, based on the square integrability of the Weyl operators, that $A*B$ is always integrable.
In fact, the integral of $A*B$ is given by
\begin{equation}\label{integral}
\int(A*B)(x) \, dx = \tr{A}\tr{B}.
\end{equation}
For extending the convolution to one merely bounded (but not integrable or trace-class) factor, we use the duality relation
\begin{equation}\label{duality}
\int f_-(x) (A*B)(x)\, dx = \tr{A_- (f*B)}=(A*(f*B))(0)
\end{equation}
which follows immediately from the definitions. This is an identity for integrable/trace class elements. If $f$ is merely bounded, the first expression still makes sense, and thus we define the convolution $f*B$ by the second expression. That is equivalent to taking the integral \eqref{convols} in the weak* sense. For $A\in\lh$ we can proceed similarly, but in this case the expressions in \eqref{convols} can also be taken literally. The correspondences \eqref{correspond} are then associated with the norm estimates
\begin{equation}\label{inequalities1}
\Vert A*S\Vert_1\leq \Vert A\Vert_1 \Vert S\Vert_1,\qquad \Vert A*S\Vert_p\leq \Vert A\Vert_p \Vert S\Vert_1 , \qquad \Vert A*S\Vert_\infty\leq \Vert A\Vert_\infty \Vert S\Vert_1.
\end{equation}

For a convenient definition for the Fourier transform we identify the dual $\widehat{\phaseS}$ with $\phaseS$ via the symplectic form $\{ (q,p),(q',p')\} = q'p-qp'$. The Fourier transform  of an integrable function is then defined as $\widehat{f}(y) = \int e^{-i\{x,y\}} f(x) \, dx $. For operators the Weyl transform $\widehat{S}(y) = \tr{SW(y)}$ for $S\in\tc$ has the equivalent role. In particular, in each of the above cases the Fourier transform maps convolutions into products, namely, $\widehat{f*g} = \widehat{f}\,\widehat{g}$, $\widehat{A*S}  = \widehat{A}\,\widehat{S}$, and $\widehat{f*S}= \widehat{f}\widehat{S}$. All of the standard results of harmonic analysis also  hold. We make explicit use of the Plancherel theorem which states that  the maps $f\mapsto \widehat{f}$ and $S\mapsto \widehat{S}$ extend to Hilbert space unitaries $\LL^2(\phaseS)\to \LL^2(\phaseS)$ and $\hs\to \LL^2(\phaseS)$.

\section{Regularity, density and injectivity}

In his classic paper Wiener \cite{Wiener} connected two conditions on the zero set of the Fourier transform of a function $f$, analogous to \zc1 and \zc2, with the property that the translates of $f$ should span an appropriate function space. It turns out that such density conditions are precisely what is needed also in the quantum case. This motivates Def.~\ref{pregular} below. Moreover, we will show that $\infty$-regularity of an operator $T\in\tc$ is equivalent to informational completeness of the phase space observable it generates (see Prop.~\ref{prop3}, condition (3.3)). The connection with zero sets will be discussed in Sect.~\ref{sec:zero}.

\begin{definition}\label{pregular}
 For $1\leq p <\infty$, we say that $T\in\schtn p$ is {\bf$p$-regular} if the linear span of $\{\alpha_x(T)\mid x\in \phaseS\}$ is dense in $\schtn p$. Similarly, a function $f\in\LL^p(\phaseS)$ is called $p$-regular, if its translates span a norm dense subspace of $\LL^p(\phaseS)$.  $1$-regular elements are just called {\bf regular}. \\
 $T\in\lh$ (resp.\ $f\in\LL^p(\phaseS)$) is called {\bf $\infty$-regular} if the span of translates is weak*-dense. \end{definition}

For $1\leq p\leq p'\leq\infty$ we have the inclusions
\begin{equation}\label{schtnInclude}
    \tc\subset\schtn{p}\subset\schtn{p'}\subset\lh,
\end{equation}
with $\|\cdot \|_\infty \leq \|\cdot\|_{p'}\leq \|\cdot\|_{p}\leq \|\cdot\|_1$; hence $p$-regularity implies $p'$-regularity.

The following three propositions characterize $p$-regularity for the quantum case. In fact, Prop.~\ref{prop2} covers the open interval $1<p<\infty$, and the endpoints $p=1$ and $p=\infty$ are stated separately as Propositions \ref{prop1} and \ref{prop3}, respectively. The reason is that the $\schtn p$ spaces for these endpoints are not reflexive, so there is an ambiguity in what one might understand under the Schatten class $\schtn\infty$, a notation we therefore avoid: should it be $\lh$, the dual of the other endpoint $\tc=\schtn1$, or should it be its predual, the space $\kh$ of compact operators, since all other $\schtn p$ consist of compact operators? So, for example, condition (2.3) with $p=1$ is (1.3) with the understanding  $\schtn\infty\mapsto\lh$, and (2.2) turns into (3.7) for $\schtn\infty\mapsto\kh$. Prop.~\ref{prop3} also has additional statements connecting weak*-density in $\lh$ with norm-density in $\kh$. We remark that this option exists also for the definition of $\infty$-regularity: it can be stated equivalently as the norm density of the translates in $\kh$.

To emphasize the common features we first state the three propositions and then give the proofs, using parallel arguments as much as possible. The spectral characterizations in terms of zero sets are given in Prop.~\ref{propZ}.

\begin{proposition}\label{prop1}
Let $T\in \tc$. Then the following conditions are  equivalent.
\begin{itemize}
\item[(1.0)] $T$ is regular.
\item[(1.1)] If $f\in \LL^\infty(\phaseS)$ and $f*T=0$, then $f=0$.
\item[(1.2)] The set $\tc*T$ is dense in $\LL^1(\phaseS)$.
\item[(1.3)] If $A\in\lh$ and $A*T=0$, then $A=0$.
\item[(1.4)] The set $\LL^1(\phaseS)*T$ is dense in $\tc$.
\item[(1.5)] $T*T$ is regular.
\item[(1.6)] For some (resp.\ all) regular $T_0\in\tc$,  $T*T_0$ is regular.
\end{itemize}
Moreover, there exists a regular operator $T\in \tc$.
\end{proposition}

\begin{proposition}\label{prop2}
Let $T\in \tc$, $1< p<\infty$, and set $q = (1-p^{-1})^{-1}$. Then the following conditions are equivalent.
\begin{itemize}
\item[(2.0)] $T$ is $p$-regular.
\item[(2.1)] If $f\in \LL^q(\phaseS)$ and $f*T=0$, then $f=0$.
\item[(2.2)] The set $\schtn p*T$ is dense in $\LL^p(\phaseS)$.
\item[(2.3)] If $A\in\schtn q$ and $A*T=0$, then $A=0$.
\item[(2.4)] The set $\LL^p(\phaseS)*T$ is dense in $\schtn p$.
\item[(2.5)] $T*T$ is $p$-regular.
\item[(2.6)] For some (resp.\ all) regular $T_0\in\tc$,  $T*T_0$ is $p$-regular.
\end{itemize}
\end{proposition}

\begin{proposition}\label{prop3}
Let $T\in \tc$. Then the following conditions are equivalent.
\begin{itemize}
\item[(3.0)] $T$ is $\infty$-regular.
\item[(3.1)] If $f\in \LL^1(\phaseS)$ and $f*T=0$, then $f=0$.
\item[(3.2)] The set $\lh*T$ is weak*-dense in $\LL^\infty(\phaseS)$.
\item[(3.3)] If $A\in\tc$ and $A*T=0$, then $A=0$.
\item[(3.4)] The set $\LL^\infty (\phaseS)*T$ is weak*-dense in $\lh$.
\item[(3.5)] $T*T$ is $\infty$-regular.
\item[(3.6)] For some (resp.\ all) regular $T_0\in\tc$,  $T*T_0$ is $\infty$-regular.
\item[(3.7)] The set $\kh*T$ is dense in $\Cn(\phaseS)$.
\item[(3.8)] The set $\Cn(\phaseS)*T$ is dense in $\kh$.
\end{itemize}
\end{proposition}

\begin{proof}\proofsec{(a.1)$\Leftrightarrow$(a.2) and (a.3)$\Leftrightarrow$(a.4) for a=1,2,3, and (3.3)$\Leftrightarrow$(3.8)}
are all based on the same basic fact concerning continuous linear operators between dual pairings of topological vector spaces \cite[Ch.~IV.2.3]{Schaefer}, i.e., in the most general setting in which the notion of adjoint makes sense: a continuous linear operator is injective if and only if its adjoint has dense range in the weak topology induced by the pairing. Indeed, the vectors in the kernel of the operator are precisely those vanishing on the range of the adjoint. In the cases at hand we have the canonical dual pairings of the Banach spaces $\duality{\LL^1(\phaseS)}{\LL^\infty(\phaseS)}$, $\duality{\LL^p(\phaseS)}{\LL^q(\phaseS)}$, $\duality\tc\lh$ and $\duality{\schtn p}{\schtn q}$. The operator involved is always written as $X\mapsto X*T$ whose adjoint, taken from \eqref{duality}, is $Y\mapsto Y*T_-$. We have omitted the minus subscripts from the statements of the theorem, because all conditions of the propositions are obviously equivalent for $T$ and for $T_-$.
We note that in the general result the natural topology in which the range is taken to be dense is the weak one induced by the pairing. However, since in a Banach space the weak closure is equal to the norm closure, the density in $\tc,\schtn p,\LL^1 (X),\LL^p(X)$ is also in norm as stated. In contrast, in Prop.~\ref{prop3}, it is the weak* topology of $\lh, \LL^\infty(X)$, i.e., the weak topology coming from the predual. The statement in the norm topology would be false. Indeed, (3.2) is always false with norm density, because all functions $A*T$, and hence their norm limits are uniformly continuous on phase space.

For proving (3.3)$\Leftrightarrow$(3.8) we take the dualities $\duality\tc\kh$ and $\duality{\LL^1(\phaseS)}{\Cn(\phaseS)}$. The latter is now not a pair of a Banach space and its dual, but still satisfies the mutual separation conditions for a duality. The operators $T*:\tc\to\LL^1(\phaseS)$, and $T_-*:\Cn(\phaseS)\to\kh$ are adjoints with respect to these pairings, therefore $T*$ is injective ($\Leftrightarrow$ (3.3)) iff the range of $T_-*$ is dense in the weak topology of $\kh$, and hence in the norm topology, which is (3.8).

\proofsec{(a.0)$\Leftrightarrow$(a.3) for a=1,2,3} Since $A*T(x) = \tr{A_-\alpha_{-x}(T)}$, also this follows immediately from the dualities $\duality\tc\lh$ and $\duality{\schtn p}{\schtn q}$.

\proofsec{(a.1)$\Rightarrow$(a.3) for a=1,2,3} follows from the associativity of convolution.
Assume (a.1), i.e. injectivity on the appropriate function class $\LL^1(X),\LL^q(X),\LL^\infty(X)$, and assume $A*T=0$ for some operator in the corresponding operator space $\tc,\schtn q,\lh$. Then, for all trace class operators $S$, we have $S*A*T=0$. But $S*A$ is in the appropriate function class, so with (a.1) we get $S*A=0$. In particular, $S*A(0)=\tr{SA_-}=0$. Since $S$ was arbitrary, $A_-=A=0$.

\proofsec{(a.5)$\Rightarrow$(a.2) and (a.6)$\Rightarrow$(a.2) for a=1,2,3} Given any $T'\in\tc$ we have $\alpha_x(T'*T) = \alpha_x(T')*T$ for all $x\in \phaseS$, so that $\{\alpha_x(T*T')\mid x\in \phaseS\}\subset \tc*T\subset \schtn p *T\subset \lh*T$ for all $1\leq p<\infty$.

\proofsec{Existence of a $T_0$ satisfying (1.2)} Let $T_0$ be the one dimensional projection onto any Gaussian wave function. Indeed, then $\widehat{T_0*T_0}(x) = \widehat{T}(x)^2$ is of (complex) Gaussian form which thus never vanishes, so we can apply the classic Wiener's approximation theorem (discussed more in the next section) to conclude that the translates of $T_0*T_0$ span $\LL^1(\phaseS)$. Since we have already proved that (1.5)$\Rightarrow$(1.2), $T_0$ satisfies (1.2).

\proofsec{(a.3)$\Rightarrow$(a.1) for a=1,2,3} is analogous to (a.1)$\Rightarrow$(a.3), with one additional idea. Assume (a.3) and $f*T=0$, with $f$ in $\LL^1(X),\LL^q(X),\LL^\infty(X)$. Then as before we get $S*f=0$ for {\it all} $S\in \tc$. To conclude the proof we choose some $T_0$ satisfying (1.2). Since $T_0*S*f=0$, we have that $g*f=0$ for the $\LL^1$-dense set of functions $g=T_0*S$. This implies $f=0$.

\proofsec{(a.2)$\Rightarrow$(a.5) and (a.2)$\Rightarrow$(a.6) for a=1,2,3}
Assuming (a.2), we can approximate any $f\in \LL^p(\phaseS)$ by $A*T$ with some $A\in \schtn p$ ($A\in \lh$ in case $p=\infty$). On the other hand, we have already proved that (a.2)$\Rightarrow$(a.1)$\Rightarrow$(a.3)$\Rightarrow$(a.0), so $T$ is $p$-regular. Hence, we can further approximate $A$ by a linear combination $\sum_{j=1}^n c_j \alpha_{x_j}(T)$; then $\sum_{j=1}^n c_j \alpha_{x_j}(T*T)$ approximates $f$ because of the $p$-norm (weak* in case $p=\infty$) continuity of $A\mapsto A*T$. This proves (a.5). If $T_0\in \tc$ is regular, it is also $p$-regular for all $1\leq p\leq \infty$, so we can also approximate $A$ by a linear combination of translates of $T_0$ instead of those of $T$; this proves (a.6).

\proofsec{(3.7)$\Leftrightarrow$(3.8)} Now note that both statements (3.7) and (3.8) hold for a regular $T_0$. Indeed, by (1.2) we can find $T_1$ so that $T_0*T_1$ is close in 1-norm to a normalized density concentrated in a small ball around the origin. Hence $\Vert{f-T_0*T_1*f}\Vert_\infty$ can be made arbitrarily small for any uniformly continuous $f\in\LL^\infty(\phaseS)$, and in particular for $f\in\Cn(\phaseS)$. Since $T_1*f\in\kh$, we conclude that $T_0*\kh\subset\Cn(\phaseS)$ is dense. Similarly, $T_0*\Cn(\phaseS)\subset\kh$ is dense.

Now assume (3.8), for some $T$. Then the set of all $A*T$ contains those with $A=f*T_0$, $f\in \Cn(\phaseS)$. But then in $A*T=(T*f)*T_0$ the first factor ranges over a dense subset of $\kh$, which by the density property of $T_0$ implies (3.7). Again the converse is completely analogous.

\end{proof}

\section{Regularity and zero sets}\label{sec:zero}

The following Proposition establishes the announced equivalence between $1,2,\infty$-regularity of a trace class operator $T$ and the ``spectral'' conditions \zc1-\zc3.  Of course, \zc1$\Rightarrow$\zc2$\Rightarrow$\zc3. These inclusions will be shown to be strict in the following section.

\begin{proposition}\label{propZ} Let $T\in \tc$. Then
\begin{itemize}
\item[(1)] $T$ is regular iff $Z(T)$ is empty \zc1.
\item[(2)] $T$ is $2$-regular iff $Z(T)$ is of measure zero \zc2.
\item[(3)] $T$ is $\infty$-regular iff $Z(T)$ has dense complement \zc3.
\end{itemize}
\end{proposition}

Since $\infty$-regularity is equivalent to informational completeness, Prop.~\ref{propZ} (3) gives the desired spectral characterization of this property. We note that in Props.~1--4, the positivity of $T$ is not required, so they are a bit more general than needed for the discussion of covariant observables. Of course, when we show later that \zc2 is not necessary, we have to be careful to construct a counterexample of a positive $T$, since the reverse implication might be true just under this additional assumption.

\begin{proof}
\proofsec{\zc1$\implies$(1.5)} Here we just refer to Wiener's approximation theorem \cite{Wiener}.

\proofsec{(1.3)$\implies$\zc1} Let $A=W(x)$ be a Weyl operator. Then $A*T$ is equal to the Weyl transform multiplied by an exponential. Hence if the Weyl transform of $T$ had a zero at $x$, we would conclude that $W(x)*T=0$ and hence, by (1.3), $W(x)=0$, which is a contradiction.

\proofsec{\zc2$\Leftrightarrow$(2.1) for $p=2$}
Assume that $Z(T)$ has measure zero, and $T*f=0$. This convolution is, in general, defined by continuous extension from $\LL^1$-functions $f$ with respect to the 2-norms. Since the Weyl transform is isometric by the quantum version of the Plancherel Theorem, this means that $\widehat T(x)\widehat f(x)=0$ for almost all $x$. But since  $\widehat T(x)\neq0$ almost everywhere, $\widehat f(x)=0$ almost everywhere, which by definition of  $\LL^2(\phaseS)$ means that $f=0$.

Conversely, assume that $Z(T)$ has positive measure. Then we can find a bounded subset $Y$, which still has positive measure. Then let $\widehat f$ be the indicator function of $Y$. This is non-zero and in $\LL^2(X)$, and hence so is $f$. On the other hand, by construction, $\widehat T(x)\widehat f(x)=0$ for all $x$, and hence by Fourier transform $T*f=0$. Hence (2.1) fails, too.

\proofsec{\zc3$\Leftrightarrow$(3.3) (Characterization of informational completeness)}
We have already noted the implication \zc3$\implies$(3.3) in the introduction: When $A*T=0$, we have $\widehat A(x)\widehat T(x)=0$ for all $x$. By assumption $\widehat T(x)\neq0$ on a dense set on which, consequently, $\widehat A(x)=0$. Since $\widehat A$ is a continuous function, it vanishes identically.

Now suppose that \zc3 does not hold, i.e. there exists an open set $\Omega\subset Z(T)$. We will construct a non-zero function $f\in\LL^1(\phaseS)$ such that $\widehat f$ has support in $\Omega$. (This will immediately be a counterexample to (3.1).) Since any regular $T_0$ satisfies \zc1 according to what we just proved above, a counterexample to (3.3) can be obtained as $A=f*T_0$.

It remains to construct a non-zero $f\in\LL^1(\phaseS)$ with ${\mathrm{supp}}\widehat f\subset\Omega$. For this we can take any sufficiently smooth function  $\widehat f$ with the required support and define $f$ by inverse Fourier transform. The Fourier transform will thus decrease faster than any desired power, and will therefore be integrable.
\end{proof}

Wiener's approximation theorem was used in a crucial way in this proof. In fact, we could have obtained the whole proposition as a corollary of classical results of classical harmonic analysis. Let us make these connections more explicit, since they will also be crucial for understanding the more subtle cases of $p$-regularity with $p\neq1,2,\infty$. Regularity statements of operators can be reduced to those of functions via the equivalences (a.0)$\Leftrightarrow$(a.5) for $a=1,2,3$ above. Similarly, the properties of zero sets are translated via the relation $Z(T)= Z(T*T)$.  Here, analogously to the introduction, we define the set $Z(f)$ for $f\in\LL^1(\phaseS)$ as the zero set of its Fourier transform $\widehat f$. With this translation, Wiener's approximation theorem \cite{Wiener} becomes Prop.~\ref{propZ}(1). For (2) we can invoke another result from \cite{Wiener}, namely that $2$-regularity is equivalent to $Z(f)$ having zero measure.  Finally, the characterization of $\infty$-regularity of functions by $Z(f)$ having dense complement is e.g. in \cite[Thm.\ (2.3)]{Edwards}.

In his classic paper Wiener already raised the question \cite[p.~93]{Wiener} about other values of $p$. This has turned out to be a subtle problem, generating a rich literature (see e.g. \cite{Segal,Beurling,Herz1957,Edwards,LevOlev}). The point we wish to make here is that the results obtained in this context can be turned directly into statements about operators using the translation principles sketched above. We begin with a statement that makes this relation more symmetric: results about operator regularity also imply classical results.

Various notions of ``smallness'' for zero sets have been considered in the literature. In the following Proposition, we introduce another one, which we call a {\it $p$-slim} set for the sake of discussion. The terminology echoes the stronger notion of $p$-thin sets of Edwards \cite{Edwards}, and a still stronger condition, sets of ``type $U^{p/(1-p)}$'' in \cite{Herz1957} (see \cite{Edwards}, particularly Thm.~2.2 for these comparisons). 

\begin{proposition}\label{propC}
For $f\in\LL^p(\phaseS)\cap \LL^1(\phaseS)$, $T\in\tc$, a regular $T_0\in\tc$, and $1\leq p\leq\infty$ we have that
$f$ is $p$-regular iff $f*T_0$ is $p$-regular, and $T$ is $p$-regular iff $T_0*T$ is $p$-regular. \\
For a subset $S\subset\Rl^2$ the following conditions are equivalent:
\begin{itemize}
\item[(1)] For any $f\in\LL^1(\phaseS)\cap\LL^p(\phaseS)$, $Z(f)\subset S$ implies that $f$ is $p$-regular.
\item[(2)] For any $T\in\tc$, $Z(T)\subset S$ implies that $T$ is $p$-regular.
\end{itemize}
We call such sets {\bf $p$-slim}.
\end{proposition}
\begin{proof}
The equivalence between the regularity of $T$ and of $T*T_0$ is just (a.0)$\Leftrightarrow$(a.6) for $a=1,2,3$ above. The corresponding statement for functions is proved in a similar fashion: $p$-regularity of $f*T_0$ implies the density of $\schtn p *f$ in $\schtn p$, which by duality implies the injectivity of $A\mapsto f*A$ on $\schtn q$, which by associativity of the convolution implies the injectivity of $g\mapsto f*g$ on $\LL^q(\phaseS)$, i.e., $p$-regularity of $f$. On the other hand, if $f$ is $p$-regular, then any $A\in \schtn p$ can be approximated by $g*T_0$ where $g$ is a linear combination of translates of $f$, so $f*T_0$ is $p$-regular. The equivalence of (1) and (2) now follows immediately, because $Z(T_0*f)=Z(f)$ and $Z(T_0*T)=Z(T)$.
\end{proof}

Clearly, $Z(T)$ being $p$-slim is a natural sufficient condition for $p$-regularity of $T$. There are several sufficient conditions for the $p$-slimness of a set $S$. As an example we give the following result, which uses Hausdorff dimension \cite{HdorffDim} as a measure of smallness. Intuitively, this describes the scaling of the number of balls needed to cover the set as a function of the radius of the balls, and is one of the standard characteristics of fractal sets. The connection with Hausdorff dimension and the closure of translates problem has been first noted by Beurling \cite{Beurling}, and extended to any dimension in \cite{Deny1950}. Since $0\leq h\leq2$ in our $2$-dimensional phase space, the bound ranges from $1$ to $2$. The proof follows by combining  \cite[Th.~4.(ii)]{Herz1957} with the fact that any set of type $U^q$, $p^{-1}+q^{-1}=1$  is also $p$-slim.

\begin{proposition} Let $2\geq p>4/(4-h)$, and let $S\subset \phaseS$ be a set of Hausdorff dimension $h$. Then $S$ is $p$-slim.
\end{proposition}

It seems that a necessary and sufficient characterization of $p$-slim sets is a hard problem. More importantly for our context, the whole research program initiated by Wiener's remark, namely to extend the clean characterizations of Prop.~\ref{propZ} to values $p\neq1,2,\infty$ has been resolved in the negative \cite{LevOlev}: information about zero sets is not in general sufficient to decide regularity. The following Proposition rephrases this result in the operator context. 

\begin{proposition} Let $1<p<2$. Then there exist $T,T'\in \tc$ such that $T$ is $p$-regular and $T'$ is not, but $Z(T) = Z(T')$.
\end{proposition}

\begin{proof} We need to extend the example established for functions of one variable in \cite[Cor.~2]{LevOlev}, to functions on phase space $\phaseS$.  Clearly, if $h,g\in L^1(\Rl)\cap \LL^p(\Rl)$, then $(q,p)\mapsto h(q)g(p)$ is $p$-regular iff the spans of translates of $h$ and $g$ are both dense in $\LL^p(\Rl)$. This follows easily by using the fact that $f\in \LL^1\cap \LL^p$ is $p$-regular iff $g*f = 0$ implies $g=0$ for all $g\in \LL^q$. Hence, we can use \cite[Cor.~2]{LevOlev} to conclude that there exist two functions $f,f'\in \LL^1(\phaseS)\cap \LL^p(\phaseS)$, such that $f$ is $p$-regular and $f'$ is not, but $Z(f)= Z(f')$. Then $T=f*T_0$ and $T'=f'*T_0$ have the stated properties if $T_0\in\tc$ is any regular operator.
\end{proof}

\section{Strict implications}

Here we show that the implications \zc1$\Rightarrow$\zc2 and \zc2$\Rightarrow$\zc3 are in fact strict.

\begin{proposition}\label{thm12}
  There is a positive trace class operator $T$ satisfying \zc2 but not \zc1.
\end{proposition}
\begin{proof}
We  take  $T=T_1 = \vert \varphi_1\rangle\langle\varphi_1 \vert$, where $\varphi_1(q) = \left( \frac{1}{\pi} \right)^{1/4} q e^{-\frac{q^2}{2}}$ is the first excited state of the harmonic oscillator. The Weyl transform is then
\begin{equation*}
    \widehat{T}_1 (q,p)= \left(\frac{1}{2}- \frac{1}{4}(q^2+p^2)\right) e^{-\frac{1}{4}(q^2+p^2)}
\end{equation*}
so clearly the zero set  $Z(T_1)=\{ (q,p) \in\phaseS \mid q^2+p^2=2 \}$ is nonempty but of measure zero.
\end{proof}

\begin{proposition}\label{thm23}
  There is a positive trace class operator $T$ satisfying \zc3 but not \zc2.
\end{proposition}

\begin{proof}
We have to construct a positive operator $T_2$ of trace one such that $Z(T_2)$ is of nonzero measure but has dense complement. We choose the form
$T_2=f*T$, with $T$ satisfying \zc1, $T\geq0$ and $\tr T=1$. Thus we reduce this to the construction of a function $f\in\LL^1(X)$, which must be positive with integral one, such that the zero set of $\widehat f$, which is equal to $Z(T_2)$, satisfies the required conditions. We further specialize this to a one-dimensional construction, by setting $f(x)=f(q,p)=\phi(q)e^{-p^2}$. Now $\phi\in \LL^1(\R)$ has to be positive, and its zero set has to meet the description. Since the zero set of $\widehat f$ now consists of infinite strips, $Z(T_2)$ constructed in this way will even have infinite measure.

As the starting point for our construction we choose a positive $\varphi\in \LL^1 (\R)$ such that also $\widehat{\varphi}$ is positive and $\widehat{\varphi} (q)\neq 0$ if and only if $q \in (-1,1)$. This is satisfied e.g.  when $\varphi = \widehat{\chi *\chi}$ where $\chi$ is the characteristic function of the interval $\left(-\tfrac{1}{2}, \tfrac{1}{2}\right)$. For each $\lambda >0 $ define
$$
\widehat{\psi}_\lambda(q) = \sum_{k\in\Z} \alpha_k \widehat{\varphi}(\lambda(q+k))
$$
where $\left(\alpha_k\right)_{k\in\Z}\in l^1(\Z)$ and $\alpha_k >0$ for all $k\in\Z$. Then $\widehat{\psi}_\lambda\in \LL^1(\R)$ and the inverse Fourier transform  gives
$$
\psi_\lambda (q) = \tfrac{1}{\lambda}  \varphi \left(\tfrac{q}{\lambda}\right)\sum_{k\in\Z} e^{-i q k} \alpha(k)
$$
so clearly $\psi_\lambda\in \LL^1 (\R)$. To ensure the positivity of $\psi_\lambda$ we need to require that the sum is positive  for all $q\in\R$. The choice $\alpha_k = 2^{-\vert k\vert}$ will work here. Finally, let $\lambda_n, \mu_n >0 $ for all $n\in\N$ and  define
$$
\widehat{\phi} (q) = \sum_{n=1}^\infty \beta_n \widehat{\psi}_{\lambda_n} \left( \mu_n q\right)
$$
where $\left(\beta_n\right)_{n\in\N}\in l^1 (\N)$ and $\beta_n>0$ for all $n\in\N$. In order to ensure that $\widehat{\phi}\in \LL^1 (\R)$ we  further assume that $\sup_{n\in\N} (\mu_n\lambda_n)^{-1} <\infty$. By construction, $\widehat{\phi}$ and $\phi$ are nonnegative, integrable functions and $\widehat{\phi}(q)=0$ if and only if $\lambda_n (\mu_n q + k)\in (-1,1)^c$ for all $n\in\N, k\in\Z$. The zero set $Z(\phi)$ of $\widehat{\phi}$ is thus
\begin{equation}
Z(\phi) = \bigcap_{n\in\N} \left( \tfrac{1}{\mu_n}\Z   + \left(-\tfrac{1}{\mu_n \lambda_n}  ,  \tfrac{1}{\mu_n\lambda_n}\right)\right)^c
\end{equation}
Clearly, we can normalize $\phi$ such that $f$ has norm one.

What remains is to show that the scaling parameters can be chosen so that $Z(\phi)$ has positive measure and dense complement. A convenient choice is now $\mu_n = 2^n$, $\lambda_n =2^{n+2}$. Then  $Z(\phi)^c$ is clearly dense since $\bigcup_{n\in\N}\frac{1}{2^n} \Z\subset Z(\phi)^c$. To show that $Z(\phi)$ is of positive measure it is sufficient to show that the measure of $Z(\phi)^c \cap [0,1]$ is strictly less than $1$.  For that purpose, note that $Z(\phi)^c \cap [0,1]= \bigcup_{n\in\N} I_n$ where
$$
I_n = \left( \tfrac{1}{2^n} \Z + \left(-\tfrac{1}{2^{2(n+1)}}, \tfrac{1}{2^{2(n+1)}} \right) \right) \cap[0,1]
$$
and the measure of $I_n$ is $\tfrac{1}{2^{n+1}}$. It follows from the subadditivity of the Lebesgue measure that the measure of $Z(\phi)^c \cap [0,1]$ is less than $\sum_{n=1}^\infty \tfrac{1}{2^{n+1}}=\tfrac{1}{2}$.  In other words, $Z(\phi) \cap [0,1]$ is of positive measure.
\end{proof}

\section{Extensions to more general phase spaces}
A general phase space $\phaseS$ can be defined as a locally compact abelian group equipped with an antisymmetric ``symplectic'' bi-character. More commonly one considers pairings $\phaseS=G\times\widehat G$, where $G$ is a locally compact abelian group, and $\widehat G$ is its dual.
The work \cite{Werner1984} was written for $G=\Rl^n$, but the extension to general $G$ is work in progress (J.S). We do not wish to enter subtleties here which are better discussed separately. Therefore, we only give a simple extension of our propositions, which is easily proved and still covers many practical cases.

\begin{proposition}
Suppose that $\phaseS=G\times\widehat G$ is a phase space such that $G$ is a finite product of copies of $\Rl$, $\Z$, the 1- torus group $\T$, and finite abelian groups. \\
Then Props.~\ref{prop1}--\ref{propZ} hold mutatis mutandis, and \zc1$\Rightarrow$\zc2$\Rightarrow$\zc3. \\
Suppose that one of the reverse implications also holds. Then $G$ is finite.
\end{proposition}

\begin{proof}For a finite cartesian product of groups written additively as $G=\bigoplus_iG_i$, we get Weyl operators which are tensor products with respect to $\hil=\bigotimes_i\hil_i$. The existence of a regular trace class operator can therefore be shown by tensoring such elements for each factor. We have seen this already for $G_i=\Rl$. For $G_i=\Z$, which is equivalent to $G_i=\T=\widehat\Z$, we can take a vector $\psi\in\ell^2(\Z)$ with $\psi(n)=a^n$ for $n\geq0$ and $\psi(n)=0$ for $n<0$. Finally, on a finite group the Hilbert space of the regular representation is also finite dimensional. For fixed $x$ the equation $\bra\psi |W(x)\psi\ket=0$ holds only on a manifold of vectors of smaller dimension, and since there are only finitely many $x$, we have that almost all pure states are regular.
This was the only specific property of $G$ needed in the proofs of the first three propositions, and the rest of the proofs is entirely parallel to the ones given above.

Suppose now that one of the implications \zc1$\Rightarrow$\zc2$\Rightarrow$\zc3 is strict for any one of the factors $G_i$. By tensoring the appropriate counterexample $T_i$ with regular elements $T_j$ we get an element $T$ whose zero set is empty/measure zero/without open sets if and only if $Z(T_i)$ has these properties. Hence in order to exclude all but finite factors, we only need to show that the inclusions are strict for $G=\Z$.
By taking $T=T_0*f$ with $T_0$ regular, and $f$ depending only on the $\Z$ coordinate, we can reduce this to finding appropriate functions on $\T$, exactly as in the proof of Prop.~\ref{thm23}. Finding $f\in\ell^1(\Z)$, whose Fourier transform has only some isolated zeros is easy. For an $f$ such that the zero set of $\widehat f$ has positive measure, but contains no open sets, we can take the same example as in the proof of Prop.~\ref{thm23}.
\end{proof}

\noindent
\textbf{Acknowledgment.} This work was partially supported by the Academy of Finland Grant No. 138135. JS was supported by the Finnish Cultural Foundation. JK was supported by Emil Aaltonen Foundation.


\end{document}